\documentclass[12pt]{article}

\newcommand\D{\displaystyle}
\newcommand\I{{\rm i}}
\newcommand\E[1]{\hbox{\rm e}^{#1}}
\newcommand\hr{\hbox{\bf R}}
\newcommand\zbar{{\bar z}}
\newcommand\mod{\,\hbox{\rm mod}\,}

\newtheorem{theorem}{\textbf{Theorem}}
\newtheorem{lemma}{\textbf{Lemma}}
\newenvironment{remark}{\textbf{Remark. }}{}
\newenvironment{proof}{\textbf{Proof }}{}

\begin{document}

\title{Darboux transformations for two dimensional elliptic
affine Toda equations}

\author{\small Zi-Xiang Zhou\\
\small School of Mathematical Sciences, Fudan University, Shanghai
200433, China\\\small  Email: zxzhou@fudan.edu.cn}

\date{}

\maketitle

\begin{abstract}
The Darboux transformations for the two dimensional elliptic affine
Toda equations corresponding to all seven infinite series of affine
Kac-Moody algebras, including $A_l^{(1)}$, $A_{2l}^{(2)}$,
$A_{2l-1}^{(2)}$, $B_l^{(1)}$, $C_l^{(1)}$, $D_l^{(1)}$ and
$D_{l+1}^{(2)}$, are presented. The Darboux transformation is
constructed uniformly for the latter six series of equations with
suitable choice of spectral parameters and the solutions of the Lax
pairs so that all the reality symmetry, cyclic symmetry and complex
orthogonal symmetry of the corresponding Lax pairs are kept
invariant. The exact solutions of all these two dimensional elliptic
affine Toda equations are obtained by using Darboux transformations.
\end{abstract}

\section{Introduction}

The two dimensional Toda equations \cite{bib:Mikhailov} are
important integrable systems. They have been studies by various
methods such as inverse scattering, Hirota method, Darboux
transformation etc.
\cite{bib:Aratyn,bib:Husinh,bib:Mackay,bib:Matveevbook,bib:Mikhailov2,bib:Sav}.
They also have important applications in Toda field theory
\cite{bib:Gomes,bib:ZhuCaldi} and in both Riemannian and affine
geometry
\cite{bib:Bobenko,bib:GHZbook,bib:McIntosh,bib:RazSav,bib:Rogersbook,bib:Sav,bib:Terng}.

For any affine Kac-Moody algebra $g$, there is a two dimensional
elliptic Toda equation \cite{bib:Kac,bib:McIntosh}
\begin{equation}
   \triangle w_k=\exp\Big(\sum_{i=1}^lc_{ik}w_i\Big)
   -\kappa_k\exp\Big(\sum_{i=1}^lc_{i0}w_i\Big)\quad
   (k=1,\cdots,l)\label{eq:Toda}
\end{equation}
where $C=(c_{ij})_{0\le i,j\le l}$ is the generalized Cartan matrix
of $g$, and $(\kappa_0,\cdots,\kappa_l)C=0$.

There are seven infinite series of affine Kac-Moody algebras:
$A_l^{(1)}$, $A_{2l}^{(2)}$, $A_{2l-1}^{(2)}$, $B_l^{(1)}$,
$C_l^{(1)}$, $D_l^{(1)}$, $D_{l+1}^{(2)}$ \cite{bib:Kac}. The
$A_l^{(1)}$ Toda equation is just the periodic Toda equation and has
been studied by a lot of authors. Especially, the Darboux
transformation was given by \cite{bib:Matveev} in a simple form. For
the two dimensional Toda equations with $g=A_l^{(1)}$,
$A_{2l}^{(2)}$ and $A_{2l-1}^{(2)}$, the Darboux transformation for
complex solutions was presented in
\cite{bib:Nirov1,bib:Nirov2,bib:Nirov3}. For the two dimensional
(real) hyperbolic affine Toda equations with $g=A_{2l}^{(2)}$,
$C_l^{(1)}$ and $D_{l+1}^{(2)}$, the binary Darboux transformation
(in integral form) was given by \cite{bib:Nimmo} and the Darboux
transformation (in differential form) was given by
\cite{bib:ZhouToda,bib:ZhouTodaCD,bib:ZhouTodaCDreal}.

In this paper, we present the construction of Darboux transformation
for a Lax pair on $\hr^2$ with a reality symmetry, a cyclic symmetry
and a complex orthogonal symmetry simultaneously. This system
contains the two dimensional elliptic Toda equations corresponding
to the affine Kac-Moody algebras $A_{2l}^{(2)}$, $A_{2l-1}^{(2)}$,
$B_l^{(1)}$, $C_l^{(1)}$, $D_l^{(1)}$ and $D_{l+1}^{(2)}$. The
degree of Darboux transformation must be high enough to keep all the
symmetries of the Lax pairs. What is more, in order to get real
solutions, the construction of Darboux transformations in elliptic
case is different from that in hyperbolic case or in complex case.
The main differences are that the distribution of their spectra are
different, and the Darboux matrix depends only on the solutions of
the Lax pair in hyperbolic or complex case but depends explicitly on
both the solutions of the Lax pair and the potentials in elliptic
case. For the two dimensional elliptic Toda equations corresponding
to $A_l^{(1)}$, there is only a reality symmetry and a cyclic
symmetry. Hence its Darboux transformation can be constructed in a
comparatively easier way.

In Section~\ref{sect:LP}, the Lax pair containing the two
dimensional elliptic $A_{2l}^{(2)}$, $A_{2l-1}^{(2)}$, $B_l^{(1)}$,
$C_l^{(1)}$, $D_l^{(1)}$ and $D_{l+1}^{(2)}$ Toda equations and its
symmetries are discussed. This leads to the symmetries of the
spectrum and the symmetries of the corresponding solutions of the
Lax pair. In Section~\ref{sect:DT}, we discuss the general
construction of Darboux transformation which keeps all the reality
symmetry, cyclic symmetry and complex orthogonal symmetry. The
simplified explicit expressions of the solutions are obtained by
Darboux transformations. In Section~\ref{sect:Egs}, the Darboux
transformation for the two dimensional elliptic $A_l^{(1)}$ Toda
equation is presented briefly, and the Darboux transformations for
the two dimensional elliptic $A_l^{(1)}$, $A_{2l}^{(2)}$,
$A_{2l-1}^{(2)}$, $B_l^{(1)}$, $C_l^{(1)}$, $D_l^{(1)}$ and
$D_{l+1}^{(2)}$ Toda equations are derived from the general results
in Section~\ref{sect:DT}.

\section{Elliptic Toda equations and their Lax pairs} \label{sect:LP}

First we present some notations.

Let $N$ and $m$ be integers with $N\ge 2$. Let $r_1,\cdots,r_N$ be
integers whose values are only $1$ or $2$ such that $r_j=2$ only if
$2j\equiv m\mod N$. Denote $M=r_1+\cdots+r_N$.

We divide any $M\times M$ matrix $A$ as $A=(A_{jk})_{1\le j,k\le N}$
where $A_{jk}$ is a matrix of order $r_j\times r_k$.  For
convenience, the subscripts of a submatrix of a matrix $A$ or a
vector $v$ are supposed to take any integer so that
$A_{j'k'}=A_{jk}$ and $v_{j'}=v_j$ if $j'\equiv j\mod N$ and
$k'\equiv k\mod N$.

Let $\omega=\exp(2\pi\I/N)$,
$\Omega=(\omega^{-j+1}I_{r_j}\delta_{jk})_{1\le j,k\le N}$. Let
$K=(K_{jk})_{1\le j,k\le N}$ be an $M\times M$ invertible real
symmetric matrix such that $K_{j,m-j}=1$ if $r_j=1$,
$K_{j,m-j}=\left(\begin{array}{cc}0&1\\1&0\end{array}\right)$ if
$r_j=2$, and $K_{jk}=0$ if $j+k\not\equiv m\mod N$. Let
$J=(J_{jk})_{1\le j,k\le N}$ be an $M\times M$ real matrix
satisfying $KJK=J^T$ and $\Omega J\Omega^{-1}=\omega J$.

Then $\Omega^N=I_M$ where $I_k$ is the $k\times k$ identity matrix,
$K$ satisfies $K^2=I_M$ and $\Omega K\Omega=\omega^{2-m}K$, and
$J_{jk}=0$ unless $k-j\equiv 1\mod N$.

\begin{remark}
Since $K$ is invertible and symmetric, $r_j=r_{m-j}$ must hold.
Hence if there is only one $j$ such that $r_j=2$, then $j$ must
satisfies $2j\equiv m\mod N$. The condition that $r_j=2$ only if
$2j\equiv m\mod N$ at the beginning of this section is a stronger
one so that the Darboux transformation constructed later will keep
all the symmetries.
\end{remark}

Let $z=(x+\I y)/2$, $\bar z=(x-\I y)/2$ be the complex coordinates
of $\hr^2$, then the Laplacian $\triangle=\partial_z\partial_{\bar
z}$. Consider the Lax pair
\begin{equation}
   \Phi_z=(\lambda J+P)\Phi,\quad \Phi_\zbar=\frac 1\lambda Q\Phi
   \label{eq:LP}
\end{equation}
where
\begin{equation}
    P=V_zV^{-1},\quad Q=VJ^TV^{-1}\label{eq:PQ}
\end{equation}
and $V(z,\bar z)$ is an $M\times M$ real diagonal matrix.

The integrability condition of (\ref{eq:LP}) is
\begin{equation}
   (V_zV^{-1})_{\bar z}+[J,VJ^TV^{-1}]=0,\label{eq:intc}
\end{equation}
which is a kind of two dimensional Toda equation when $V$ has
certain symmetries.

We will consider $V$ satisfying reality symmetry $\bar V=V$, cyclic
symmetry $\Omega V\Omega^{-1}=V$ and complex orthogonal symmetry
$V^TKV=K$ with respect to the indefinite metric given by the real
symmetric matrix $K$, so that (\ref{eq:intc}) will lead to the two
dimensional elliptic $A_{2l}^{(2)}$, $A_{2l-1}^{(2)}$, $B_l^{(1)}$,
$C_l^{(1)}$, $D_l^{(1)}$ and $D_{l+1}^{(2)}$ Toda equations. Without
the complex orthogonal symmetry, (\ref{eq:intc}) will lead to the
two dimensional elliptic $A_l^{(1)}$ Toda equation.

By direct computation, we have
\begin{lemma}
If the diagonal matrix $V(z,\bar z)$ satisfies
\begin{equation}
   \bar V=V,\quad \Omega V\Omega^{-1}=V,\quad
   V^TKV=K,\label{eq:Vred}
\end{equation}
then $P=V_zV^{-1}$ and $Q=VJ^TV^{-1}$ satisfy
\begin{equation}
   \begin{array}{l}
   \Omega P\Omega ^{-1}=P,\quad \Omega Q\Omega^{-1}=\omega^{-1} Q,\\
   KPK=-P^T,\quad KQK=Q^T.\label{eq:PQred}
   \end{array}
\end{equation}
\end{lemma}

\begin{remark}
If $\Omega$ is changed to $\omega\Omega$, then the relation $\Omega
K\Omega=\omega^{2-m}K$ is changed to $\Omega
K\Omega=\omega^{2-(m+2)}K$, but $\Omega J\Omega^{-1}=\omega J$ and
$\Omega V\Omega^{-1}=V$ are unchanged. This means that the system is
invariant if $m$ is changed to $m+2$ and any subscript $j$ is
changed to $j+1$. Therefore, it is only necessary to choose $m=1$ or
$m=2$ when $N$ is even and $m=2$ when $N$ is odd. Other systems are
equivalent to these ones.
\end{remark}

The symmetries of the Lax pair lead to the symmetries of its
solutions. This fact is shown in the following lemma and will be
used in constructing Darboux transformations in the next section.

\begin{lemma}\label{lemma:sys}
Suppose $V$ satisfies (\ref{eq:Vred}), $\Phi$ is a solution of the
Lax pair (\ref{eq:LP}) with $\lambda=\lambda_0$, then the following
facts hold.

(i) $\Omega\Phi$ is a solution of (\ref{eq:LP}) with
$\lambda=\omega\lambda_0$.

(ii) $KV^{-1}\bar\Phi$ is a solution of (\ref{eq:LP}) with
$\lambda=1/\bar\lambda_0$.

(iii) $\Psi=K\Phi$ is a solution of the adjoint Lax pair with
$\lambda=-\lambda_0$:
\begin{equation}
   \Psi_z=-(-\lambda_0 J^T+P^T)\Psi,\quad
   \Psi_\zbar=-\frac 1{-\lambda_0}Q^T\Psi.
   \label{eq:LPadj-}
\end{equation}

(iv) When $N$ is even, $\Psi=K\Omega^{N/2}\Phi$ is a solution of the
adjoint Lax pair with $\lambda=\lambda_0$:
\begin{equation}
   \Psi_z=-(\lambda_0 J^T+P^T)\Psi,\quad
   \Psi_\zbar=-\frac 1{\lambda_0}Q^T\Psi.
   \label{eq:LPadj}
\end{equation}
Hence $\Phi^TK\Omega^{N/2}\Phi$ is a constant when $N$ is even.
\end{lemma}

\begin{proof}
(i) follows from
\begin{equation}
   \begin{array}{l}
   (\Omega\Phi)_z=\Omega(\lambda_0 J+V_zV^{-1})\Phi
   =(\omega\lambda_0 J+V_zV^{-1})\Omega\Phi,\\
   \D(\Omega\Phi)_\zbar=\frac 1{\lambda_0}\Omega VJ^TV^{-1}\Phi
   =\frac 1{\omega\lambda_0}VJ^TV^{-1}\Omega\Phi.
   \end{array}
\end{equation}

(ii) follows from
\begin{equation}
   \begin{array}{l}
   \D(KV^{-1}\bar\Phi)_z=\frac 1{\bar\lambda_0}KV^{-1}VJ^TV^{-1}\bar\Phi
   -KV^{-1}V_zV^{-1}\bar\Phi\\
   \qquad\D=\Big(\frac 1{\bar\lambda_0}J+V_zV^{-1}\Big)KV^{-1}\bar\Phi,\\
   \D(KV^{-1}\bar\Phi)_\zbar
   =KV^{-1}(\bar\lambda_0
   J+V_\zbar V^{-1})\bar\Phi-KV^{-1}V_\zbar V^{-1}\bar\Phi\\
   \qquad\D=\bar\lambda_0VJ^TV^{-1}KV^{-1}\bar\Phi.\\
   \end{array}
\end{equation}

(iii) follows from
\begin{equation}
   \begin{array}{l}
   \D K\Phi_z=K(\lambda_0 J+V_zV^{-1})\Phi=-(-\lambda_0 J^T+(V_zV^{-1})^T)K\Phi,\\
   \D K\Phi_\zbar=\frac 1{\lambda_0}KVJ^TV^{-1}\Phi
   =-\frac 1{-\lambda_0}(VJ^TV^{-1})^TK\Phi.
   \end{array}
\end{equation}

(iv) follows from (i) and (iii). The lemma is proved.
\end{proof}

\section{Darboux transformation}\label{sect:DT}

A matrix
\begin{equation}
   T(z,\bar z,\lambda)=\sum_{j=0}^LT_j(z,\bar z)\lambda^{L-j},\quad
   T_0=I\label{eq:Gpoly}
\end{equation}
is called a Darboux matrix of degree $L$ for the Lax pair
(\ref{eq:LP}) if there exists a diagonal matrix $\widetilde
V(z,\zbar)$ satisfying
\begin{equation}
   \bar{\widetilde V}=\widetilde V,\quad
   \Omega \widetilde V\Omega^{-1}=\widetilde V,\quad
   \widetilde V^TK\widetilde V=K,\label{eq:Vred2}
\end{equation}
such that for any solution $\Phi$ of (\ref{eq:LP}),
$\widetilde\Phi=T\Phi$ satisfies
\begin{equation}
   \widetilde\Phi_z=(\lambda J+\widetilde P)\widetilde\Phi,\quad
   \widetilde\Phi_\zbar=\frac 1\lambda \widetilde Q\widetilde\Phi
   \label{eq:LP2}
\end{equation}
with
\begin{equation}
    \widetilde P=\widetilde V_z\widetilde V^{-1},\quad
    \widetilde Q=\widetilde VJ^T\widetilde V^{-1}.
\end{equation}

If a Darboux matrix is constructed, new solutions of the two
dimensional elliptic affine Toda equations can be obtained from a
known one.

When the symmetries in (\ref{eq:Vred2}) are not considered, the
Darboux transformation can be constructed according to
\cite{bib:Zak} (c.f. also \cite{bib:Cis}) as follows.

Let $\lambda_1,\cdots,\lambda_L,\hat\lambda_1,\cdots,\hat\lambda_L$
be distinct complex constants. Let $H_j$ be a solution of the Lax
pair with $\lambda=\lambda_j$, $\hat H_j$ be a solution of the
adjoint Lax pair with $\lambda=\hat \lambda_j$. Let
$\D\Gamma_{jk}=\frac{\hat H_j^TH_k}{\lambda_k-\hat\lambda_j}$
$(j,k=1,\cdots,L)$, $\check\Gamma=\Gamma^{-1}$. Then
\begin{equation}
   T(\lambda)=\prod_{i=1}^L(\lambda-\hat\lambda_i)
   \Big(I-\sum_{j,k=1}^L\frac{H_j\check\Gamma_{jk}\hat
   H_k^T}{\lambda-\hat\lambda_k}\Big)
\end{equation}
is a Darboux matrix without considering symmetries.

However, for our problem, the symmetries in (\ref{eq:Vred2}) must be
satisfied. When $N$ is odd, the construction of Darboux
transformation is comparatively easier. However, when $N$ is even,
the following problem is encountered. In fact, (i) and (iii) of
Lemma~\ref{lemma:sys} imply that if $\lambda_0$ is an eigenvalue of
the Lax pair~(\ref{eq:LP}), $\omega^{N/2}\lambda_0=-\lambda_0$ is an
eigenvalue of the adjoint Lax pair (\ref{eq:LPadj}). If we choose
$\lambda_j=\omega^{j-1}\mu$,
$\lambda_{N+j}=\omega^{j-1}\bar\mu^{-1}$,
$\hat\lambda_j=-\omega^{j-1}\mu$,
$\hat\lambda_{N+j}=-\omega^{j-1}\bar\mu^{-1}$ $(j=1,\cdots,N)$ with
$\mu\ne 0$ and $|\mu|\ne 1$ as introduced by Lemma~\ref{lemma:sys},
then
$\lambda_1,\cdots,\lambda_{2N},\hat\lambda_n,\cdots,\hat\lambda_{2N}$
are not distinct, so $\Gamma$ will not exist, let alone
$T(\lambda)$. On the other hand, if we choose
$\lambda_j=\omega^{j-1}\mu$, $\hat\lambda_j=-\omega^{j-1}\mu$
$(j=1,\cdots,N)$ with $\mu\ne 0$ and $|\mu|\ne 1$, then
$\lambda_1,\cdots,\lambda_{N},\hat\lambda_n,\cdots,\hat\lambda_{N}$
are distinct. However, in order to keep the complex orthogonal
symmetry, the solution $H_j$ cannot be chosen as column solution,
but $M\times M/2$ matrix solution, which makes the result
complicated when $M$ is even and impossible when $M$ is odd. In
order to solve this problem, a limit process is needed.

No matter $N$ is odd or even, the Darboux transformation can be
constructed uniformly as follows.

Let $\Phi(z,\bar z,\lambda)$ be a column solution of the Lax pair
(\ref{eq:LP}). When $N$ is even, $\Phi(z,\bar z,\lambda)$ should
satisfy an extra condition that
\begin{equation}
   \Phi(z_0,\bar z_0,\lambda)^TK\Omega^{N/2}\Phi(z_0,\bar z_0,\lambda)=0
\end{equation}
holds for all $\lambda$ and certain point $(z_0,\bar z_0)$. Then
according to (iv) of Lemma~\ref{lemma:sys},
\begin{equation}
   \Phi(z,\bar z,\lambda)^TK\Omega^{N/2}\Phi(z,\bar z,\lambda)=0
\end{equation}
holds identically.

Let
\begin{equation}
   \begin{array}{l}
   \lambda_j(\lambda)=\omega^{j-1}\lambda, \quad
   \Phi_j(z,\bar z,\lambda)=\Omega^{j-1}\Phi(z,\bar z,\lambda),\\
   \lambda_{N+j}(\lambda)=\omega^{j-1}\bar\lambda^{-1},\quad
   \Phi_{N+j}(z,\bar z,\lambda)=KV^{-1}\Omega^{-j+1}\overline{\Phi(z,\bar z,\lambda)}
   \end{array}
\end{equation}
$(j=1,2,\cdots,N)$. According to (i) and (ii) of
Lemma~\ref{lemma:sys}, each $\Phi_j(z,\bar z,\lambda)$ is a solution
of (\ref{eq:LP}) with $\lambda$ replaced by $\lambda_j(\lambda)$
$(j=1,\cdots,2N)$, and $\Phi_{j+1}=\Omega\Phi_j$,
$\Phi_{N+j+1}=\omega^{m-2}\Omega\Phi_{N+j}$ for $j=1,2,\cdots,N-1$.
Moreover, (iv) of Lemma~\ref{lemma:sys} implies that
$\Phi_j^TK\Omega^{N/2}\Phi_j=0$ holds for $j=1,2,\cdots,2N$ when $N$
is even.

Let $\mu$ be a non-zero complex number with $|\mu|\ne 1$. Denote
\begin{equation}
   H(z,\bar z)=\Phi(z,\bar z,\mu),\quad
   H'(z,\bar z)=\frac{\partial\Phi}{\partial\lambda}(z,\bar z,\lambda)\Big|_{\lambda=\mu}.
\end{equation}
Then
\begin{equation}
   H^TK\Omega^{N/2}H=0,\quad
   H^{\prime\,T}K\Omega^{N/2}H=-HK\Omega^{N/2}H'
   \label{eq:Hconstr}
\end{equation}
when $N$ is even. Note that the relations in (\ref{eq:Hconstr}) are
identities without any conditions when $N$ is even and $m$ is odd,
since $(K\Omega^{N/2})^T=(-1)^{2-m}K\Omega^{N/2}$.

Let
\begin{equation}
   \begin{array}{l}
   \mu_j=\omega^{j-1}\mu,\quad H_j=\Omega^{j-1}H,\\
   \mu_{N+j}=\omega^{j-1}\bar\mu^{-1},\quad H_{N+j}=KV^{-1}\Omega^{-j+1}\bar H
   \end{array}
\end{equation}
$(j=1,2,\cdots,N)$. Moreover, let $H_j'=\Omega^{j-1}H'$,
$H_{N+j}'=KV^{-1}\Omega^{-j+1}\bar H'$ $(j=1,2,\cdots,N)$ when $N$
is even.

Let $\Gamma=(\Gamma_{jk})_{2N\times 2N}$ where
\begin{equation}
   \Gamma_{jk}(z,\zbar)=\lim_{\lambda\to\mu}
   \frac{\Phi_j^T(z,\bar z,\lambda)K\Phi_k(z,\bar z,\mu)}{\lambda_j(\lambda)+\mu_k}.
   \label{eq:Gamma}
\end{equation}
By direct computation, we have
\begin{lemma}
If $N$ is odd, or $N$ is even and $\mu_j+\mu_k\ne 0$, then
\begin{equation}
   \Gamma_{jk}=\frac{H_j^TKH_k}{\mu_j+\mu_k}.\label{eq:GammaEven1}
\end{equation}
If $N=2n$ is even and $\mu_j+\mu_k=0$, then $k=j+n$ $(1\le j\le n)$
or $j=k+n$ $(1\le k\le n)$, and
\begin{equation}
   \begin{array}{l}
   \D\Gamma_{j,n+j}=\omega^{-j+1}H_j^{\prime\,T}KH_{n+j},\\
   \D\Gamma_{n+j,j}=-\omega^{-j+1}H_{n+j}^{\prime\,T}KH_j,\\
   \D\Gamma_{N+j,N+n+j}=-\omega^{-j+1}\bar\mu^2H_{N+j}^{\prime
   T}KH_{N+n+j},\\
   \D\Gamma_{N+n+j,N+j}=\omega^{-j+1}\bar\mu^2H_{N+n+j}^{\prime\,T}KH_{N+j}
   \end{array}
   \label{eq:GammaEven2}
\end{equation}
for $j=1,2,\cdots,n$. Moreover, for all $j,k=1,\cdots,2N$,
$\Gamma_{jk}$'s satisfy
\begin{equation}
   \Gamma_{jk}^T=\Gamma_{kj}, \label{eq:GammaSym}
\end{equation}
and
\begin{equation}
   \Gamma_{jk,z}=H_j^TKJH_k,\quad
   \Gamma_{jk,\bar z}=\frac 1{\mu_j\mu_k}H_j^TKQH_k.\label{eq:Gammaeq}
\end{equation}
\end{lemma}

Denote $\check\Gamma=\Gamma^{-1}$ when $\Gamma$ is invertible. Let
\begin{equation}
   T(\lambda)=\prod_{i=1}^{2N}(\lambda+\mu_i)
   \Big(I-\sum_{j,k=1}^{2N}\frac{H_j\check\Gamma_{jk}H_k^TK}{\lambda+\mu_k}\Big),
   \label{eq:G}
\end{equation}
then $T(\lambda)$ is of form (\ref{eq:Gpoly}) with $L=2N$.

It can be checked directly that
\begin{equation}
   T(\lambda)^{-1}=\prod_{i=1}^{2N}(\lambda+\mu_i)^{-1}
   \Big(I+\sum_{j,k=1}^{2N}\frac{H_j\check\Gamma_{jk}H_k^TK}{\lambda-\mu_j}\Big).
   \label{eq:Ginv}
\end{equation}
Moreover, from (\ref{eq:G}),
\begin{equation}
   \begin{array}{l}
   \D T_1=\sum_{i=1}^{2N}\mu_i-\sum_{j,k=1}^{2N}
   H_j\check\Gamma_{jk}H_k^TK,\\
   \D T_{2N}=\prod_{i=1}^{2N}\mu_i
   \Big(I-\sum_{j,k=1}^{2N}\frac{H_j\check\Gamma_{jk}H_k^TK}{\mu_k}\Big).
   \end{array}\label{eq:TL}
\end{equation}

\begin{lemma}\label{lemma:DTPQV}
Suppose $T(\lambda)$ is given by (\ref{eq:G}), then (\ref{eq:LP2})
holds where
\begin{equation}
   \widetilde P=P-[J,T_1],\quad \widetilde
   Q=T_{2N}QT_{2N}^{-1}.\label{eq:DTPQ}
\end{equation}
Moreover, if the diagonal matrix $\widetilde V$ satisfies
\begin{equation}
   \widetilde V=\rho^{-1}T_{2N}V\label{eq:DTV}
\end{equation}
where $\rho$ is a non-zero complex constant, then $\widetilde
P=\widetilde V_z\widetilde V^{-1}$ and $\widetilde Q=\widetilde
VJ^T\widetilde V^{-1}$ satisfy (\ref{eq:DTPQ}).
\end{lemma}

\begin{proof}
With (\ref{eq:Gammaeq}) and (\ref{eq:DTPQ}), it can be verified by
direct computation that
\begin{equation}
   \begin{array}{l}
   (\lambda J+\widetilde P)T
   -T(\lambda J+P)-T_z=0,\\
   \lambda^{-1}\widetilde QT
   -\lambda^{-1} T Q-T_{\bar z}=0
   \end{array}\label{eq:DTgauge}
\end{equation}
hold for all $\lambda$, which are equivalent to (\ref{eq:LP2}).

Clearly, $\widetilde VJ^T\widetilde V^{-1}=T_{2N}QT_{2N}^{-1}$ is
true. It is only necessary to prove that $\widetilde V_z\widetilde
V^{-1}=P-[J,T_1]$ holds. Taking $\lambda=0$ in
\begin{equation}
   T_z+T(\lambda J+P)=(\lambda J+P-[J,T_1])T
\end{equation}
which is the first equation of (\ref{eq:DTgauge}), we get
\begin{equation}
   T_{{2N},z}+T_{2N}P=(P-[J,T_1])T_{2N}.
\end{equation}
On the other hand,
\begin{equation}
   T_{{2N},z}+T_{2N}P=\widetilde V_z\widetilde V^{-1} T_{2N}
\end{equation}
holds by differentiating (\ref{eq:DTV}). Hence $\widetilde
V_z\widetilde V^{-1}=P-[J,T_1]$. The lemma is proved.
\end{proof}

\begin{theorem}\label{thm:DT_sys}
Suppose $V$ satisfies $\bar V=V$, $\Omega V\Omega^{-1}=V$ and
$V^TKV=K$, then $T(\lambda)$ satisfies
\begin{eqnarray}
   &&\D T(-\lambda)^TKT(\lambda)=(\lambda^{2N}-\mu^{2N})(\lambda^{2N}-\bar\mu^{-2N})K,
   \label{eq:DT_adj}\\
   &&\D\Omega
   T(\lambda)\Omega^{-1}=T(\omega\lambda),\label{eq:DT_cyc}\\
   &&\D\lambda^{2N}\overline{T(\bar\lambda^{-1})}=KV^{-1}T_{2N}^{-1}T(\lambda)VK.
   \label{eq:DT_rea}
\end{eqnarray}

Especially, $T_{2N}$ satisfies the relations
$T_{2N}^TKT_{2N}=(\mu/\bar\mu)^{2N}K$, $\Omega
T_{2N}\Omega^{-1}=T_{2N}$ and
$T_{2N}^*=(\bar\mu/\mu)^{2N}V^{-1}T_{2N}V$ where ${}^*$ refers to
Hermitian conjugation. Therefore, $(T_{2N})_{jk}\ne 0$ only if
$j=k$, and the diagonal part of $(\bar\mu/\mu)^N(T_{2N})$ is real.

Moreover, if $(T_{2N})_{jj}$ is diagonal whenever $r_j>1$, then
$(\bar\mu/\mu)^NT_{2N}$ is a real diagonal matrix, and $\widetilde
V=(\bar\mu/\mu)^NT_{2N}V$ satisfies $\bar{\widetilde V}=\widetilde
V$, $\Omega\widetilde V\Omega^{-1}=\widetilde V$ and $\widetilde
V^TK\widetilde V=K$.
\end{theorem}

\begin{proof}
With $\check\Gamma_{jk}^T=\check\Gamma_{kj}$, (\ref{eq:DT_adj})
follows directly from (\ref{eq:G}), (\ref{eq:Ginv}) and the fact
$\D\prod_{i=1}^{2N}(\mu_i^2-\lambda^2)
=(\lambda^{2N}-\mu^{2N})(\lambda^{2N}-\bar\mu^{-2N})$.

Denote
\begin{equation}
   \Gamma=\left(\begin{array}{cc}A&B\\C&D\end{array}\right)
\end{equation}
where $A$, $B$, $C$, $D$ are $N\times N$ matrices. For
$j,k=1,\cdots,2N$,
\begin{equation}
   \begin{array}{l}
   \D A_{jk}=\frac{H_j^TKH_k}{\mu_j+\mu_k}
   =\frac{H^T\Omega^{j-1}K\Omega^{k-1}H}{\omega^{j-1}\mu+\omega^{k-1}\mu}
   =\omega^{(1-m)(j-1)}\frac{H^TK\Omega^{k-j}H}{(1+\omega^{k-j})\mu}\\
   \qquad(\hbox{if }\omega^{k-j}\ne -1),\\
   \D B_{jk}=\frac{H_j^TKH_{N+k}}{\mu_j+\mu_{N+k}}
   =\frac{H^T\Omega^{j-1}V^{-1}\Omega^{-k+1}\bar H}{\omega^{j-1}\mu+\omega^{k-1}\bar\mu^{-1}}
   =\omega^{-(j-1)}\frac{H^TV^{-1}\Omega^{j-k}\bar H}{\mu+\omega^{k-j}\bar\mu^{-1}},\\
   \D C_{jk}=\frac{H_{N+j}^TKH_k}{\mu_{N+j}+\mu_k}
   =\frac{H^*\Omega^{-j+1}V^{-1}\Omega^{k-1}H}{\omega^{j-1}\bar\mu^{-1}+\omega^{k-1}\mu}
   =\omega^{-(j-1)}\frac{H^* V^{-1}\Omega^{k-j}H}{\bar\mu^{-1}+\omega^{k-j}\mu},\\
   \D D_{jk}=\frac{H_{N+j}^TKH_{N+k}}{\mu_{N+j}+\mu_{N+k}}
   =\frac{H^*\Omega^{-j+1}V^{-1}KV^{-1}\Omega^{-k+1}\bar H}
   {\omega^{j-1}\bar\mu^{-1}+\omega^{k-1}\bar\mu^{-1}}\\
   \qquad\D=\omega^{(m-3)(j-1)}\frac{H^* K\Omega^{j-k}\bar
   H}{(1+\omega^{k-j})\bar\mu^{-1}}
   \qquad(\hbox{if }\omega^{k-j}\ne -1).
   \end{array}\label{ABCD1}
\end{equation}
$\omega^{k-j}=-1$ occurs only when $N=2n$ is even and $j-k\equiv n\mod
N$. In this case,
\begin{equation}
   \begin{array}{l}
   \D A_{j,j+\varepsilon n}=\omega^{(1-m)(j-1)}H^{\prime\,T}K\Omega^nH,\\
   \D D_{j,j+\varepsilon n}=-\omega^{(m-3)(j-1)}\bar\mu^2H^{\prime\,*}K\Omega^n\bar H
   \end{array}\label{ABCD2}
\end{equation}
where $\varepsilon=1$ for $j=1,\cdots,n$ and $\varepsilon=-1$ for $j=n+1,\cdots,2n$.

Hence $\Gamma$ satisfies
\begin{equation}
   \begin{array}{l}
   \left(\begin{array}{cc}E\\&E\end{array}\right)\Gamma
   \left(\begin{array}{cc}E\\&E\end{array}\right)^{-1}
   =\omega^{-1}\left(\begin{array}{cc}1\\&\omega^{m-2}\end{array}\right)
   \Gamma\left(\begin{array}{cc}\omega^{2-m}\\&1\end{array}\right)
   \end{array}
\end{equation}
where $\D E=(\delta_{j,k-1})_{1\le j,k\le N}$. Denote
$\D\Gamma^{-1}=\left(\begin{array}{cc}\check A&\check B\\\check
C&\check D\end{array}\right)$ where $\check A$, $\check B$, $\check
C$, $\check D$ are $N\times N$ matrices, then
\begin{equation}
   \left(\begin{array}{cc}E\\&E\end{array}\right)\Gamma^{-1}
   \left(\begin{array}{cc}E\\&E\end{array}\right)^{-1}
   =\omega\left(\begin{array}{cc}\omega^{m-2}\\&1\end{array}\right)
   \Gamma^{-1}\left(\begin{array}{cc}1\\&\omega^{2-m}\end{array}\right).
\end{equation}
Hence for $j,k=1,\cdots,N$,
\begin{equation}
   \begin{array}{l}
   \check A_{j-1,k-1}=\omega^{1-m}\check A_{jk},\quad
   \check B_{j-1,k-1}=\omega^{-1}\check B_{jk},\\
   \check C_{j-1,k-1}=\omega^{-1}\check C_{jk},\quad
   \check D_{j-1,k-1}=\omega^{m-3}\check D_{jk}.
   \end{array}\label{eq:Gammainv}
\end{equation}
On the other hand, by the definition of $H_j$'s,
\begin{equation}
   \begin{array}{l}
   \Omega H_j=H_{j+1},\quad \Omega H_{N+j}=\omega^{2-m}H_{N+j+1},\\
   H_k^TK\Omega^{-1}=\omega^{m-2}H_{k+1}^TK,\quad
   H_{N+k}^TK\Omega^{-1}=H_{N+k+1}^TK
   \end{array}\label{eq:OmegaH}
\end{equation}
for $j,k=1,\cdots,N-1$. (\ref{eq:Gammainv}) and (\ref{eq:OmegaH})
lead to
\begin{equation}
   \begin{array}{l}
   \D \Omega T(\lambda)\Omega^{-1}=\prod_{i=1}^{2N}(\lambda+\mu_i)
   \Omega\Big(I-\sum_{j,k=1}^N\frac{H_j\check A_{jk}H_k^TK}{\lambda+\mu_k}
   -\sum_{j,k=1}^N\frac{H_j\check B_{jk}H_{N+k}^TK}{\lambda+\mu_{N+k}}\\
   \D\quad\qquad-\sum_{j,k=1}^N\frac{H_{N+j}\check C_{jk}H_k^TK}{\lambda+\mu_k}
   -\sum_{j,k=1}^N\frac{H_{N+j}\check D_{jk}H_{N+k}^TK}{\lambda+\mu_{N+k}}
   \Big)\Omega^{-1}\\
   \D=\prod_{i=1}^{2N}(\lambda+\mu_i)
   \Big(I-\omega^{m-2}\sum_{j,k=1}^N
   \frac{H_{j+1}\check A_{jk}H_{k+1}^TK}{\lambda+\mu_k}
   -\sum_{j,k=1}^N\frac{H_{j+1}\check B_{jk}H_{N+k+1}^TK}
   {\lambda+\mu_{N+k}}\\
   \D\quad\qquad-\sum_{j,k=1}^N\frac{H_{N+j+1}\check C_{jk}H_{k+1}^TK}
   {\lambda+\mu_k}
   -\omega^{2-m}\sum_{j,k=1}^N
   \frac{H_{N+j+1}\check D_{jk}H_{N+k+1}^TK}
   {\lambda+\mu_{N+k}}
   \Big)\\
   \D=\prod_{i=1}^{2N}(\omega\lambda+\mu_i)
   \Big(I-\omega^{m-1}\sum_{j,k=1}^N
   \frac{H_j\check A_{j-1,k-1}H_k^TK}{\omega\lambda+\mu_k}
   -\omega\!\!\sum_{j,k=1}^N\frac{H_j\check B_{j-1,k-1}H_{N+k}^TK}
   {\omega\lambda+\mu_{N+k}}\\
   \D\quad\qquad-\omega\sum_{j,k=1}^N\frac{H_{N+j}\check C_{j-1,k-1}H_k^TK}
   {\omega\lambda+\mu_k}
   -\omega^{3-m}\sum_{j,k=1}^N
   \frac{H_{N+j}\check D_{j-1,k-1}H_{N+k}^TK}
   {\omega\lambda+\mu_{N+k}}
   \Big)\\
   =T(\omega\lambda).
   \end{array}
\end{equation}
This proves (\ref{eq:DT_cyc}).

By (\ref{ABCD1}) and (\ref{ABCD2}),
\begin{equation}
   \begin{array}{l}
   \bar A_{jk}=\omega^{j+k-2}\bar\mu^{-2}D_{jk},\quad
   \bar B_{jk}=\omega^{j+k-2}\bar\mu^{-1}\mu C_{jk},\\
   \bar C_{jk}=\omega^{j+k-2}\bar\mu^{-1}\mu B_{jk},\quad
   \bar D_{jk}=\omega^{j+k-2}\mu^2A_{jk},
   \end{array}
\end{equation}
i.e.
\begin{equation}
   \begin{array}{l}
   \bar A=\bar\mu^{-2}\Omega^{-1} D\Omega^{-1},\quad
   \bar B=\bar\mu^{-1}\mu\Omega^{-1} C\Omega^{-1},\quad\\
   \bar C=\bar\mu^{-1}\mu\Omega^{-1} B\Omega^{-1},\quad
   \bar D=\mu^2\Omega^{-1} A\Omega^{-1}.
   \end{array}
\end{equation}
Hence
\begin{equation}
   \bar\Gamma=\left(\begin{array}{cc}&\bar\mu^{-1}\Omega^{-1}\\\mu\Omega^{-1}\end{array}\right)
   \Gamma\left(\begin{array}{cc}&\mu\Omega^{-1}\\\bar\mu^{-1}\Omega^{-1}\end{array}\right),
\end{equation}
\begin{equation}
   \bar\Gamma^{-1}=\left(\begin{array}{cc}&\bar\mu\Omega\\\mu^{-1}\Omega\end{array}\right)
   \Gamma^{-1}
   \left(\begin{array}{cc}&\mu^{-1}\Omega\\\bar\mu\Omega\end{array}\right).
\end{equation}
Written in components,
\begin{equation}
   \begin{array}{l}
   \bar{\check A}_{jk}=\omega^{-j-k+2}\bar\mu^2{\check D}_{jk},\quad
   \bar{\check B}_{jk}=\omega^{-j-k+2}\mu^{-1}\bar\mu{\check C}_{jk},\\
   \bar{\check C}_{jk}=\omega^{-j-k+2}\mu^{-1}\bar\mu{\check B}_{jk},\quad
   \bar{\check D}_{jk}=\omega^{-j-k+2}\mu^{-2}{\check A}_{jk},
   \end{array}
\end{equation}
i.e.
\begin{equation}
   \bar{\check\Gamma}_{jk}=\frac
   1{\mu_{N+j}\mu_{N+k}}\check\Gamma_{N+j,N+k}.
   \label{eq:Gammabar}
\end{equation}

By the definition of $H_j$'s,
\begin{equation}
   \begin{array}{l}
   \bar H_j=KV^{-1}H_{N+j},\quad
   H_k^* K=H_{N+k}^TV^{-1}\quad(j,k=1,\cdots,2N).
   \end{array}\label{eq:Omegabar}
\end{equation}

By (\ref{eq:Gammabar}), (\ref{eq:Omegabar}) and the fact
$\mu_{N+k}\bar\mu_k=1$ for $k=1,\cdots,2N$,
\begin{equation}
   \begin{array}{l}
   \D\lambda^{2N}T_{2N}KV^{-1}\overline{T(\bar\lambda^{-1})}VK\\
   \D=\lambda^{2N}\prod_{i=1}^{2N}(\lambda^{-1}+\bar\mu_i)T_{2N}KV^{-1}
   \Big(I-\sum_{j,k=1}^{2N}\frac{\bar H_j\bar{\check\Gamma}_{jk}H_k^* K}
   {\lambda^{-1}+\bar\mu_k}\Big)VK\\
   \D=\prod_{i=1}^{2N}(1+\bar\mu_i\lambda)
   T_{2N}\Big(I-\sum_{j,k=1}^{2N}\frac{H_{N+j}\check\Gamma_{N+j,N+k}H_{N+k}^TK}
   {\mu_{N+j}\mu_{N+k}(\lambda^{-1}+\bar\mu_k)}
   \Big)\\
   \D=\prod_{i=1}^{2N}(\lambda+\mu_i)
   \Big(I-\sum_{a,b=1}^{2N}\frac{H_{a}\check\Gamma_{ab}H_{b}^TK}
   {\mu_{b}}\Big)
   \Big(I-\sum_{j,k=1}^{2N}\frac{\lambda H_{j}\check\Gamma_{jk}H_{k}^TK}
   {\mu_{j}(\lambda+\mu_k)}
   \Big)\\
   \D=\prod_{i=1}^{2N}(\lambda+\mu_i)
   \Big(I-\sum_{j,k=1}^{2N}\frac{H_{j}\check\Gamma_{jk}H_{k}^TK}
   {\lambda+\mu_k}
   \Big)=T(\lambda).
   \end{array}
\end{equation}
Hence (\ref{eq:DT_rea}) holds.

Now we consider the symmetries of $T_{2N}$.

$T_{2N}^TKT_{2N}=(\mu/\bar\mu)^{2N}K$ and $\Omega
T_{2N}\Omega^{-1}=T_{2N}$ are derived by setting $\lambda=0$ in
(\ref{eq:DT_adj}) and (\ref{eq:DT_cyc}). Taking the coefficient of
$\lambda^{2N}$ in (\ref{eq:DT_rea}), we get
$\bar{T}_{2N}=KV^{-1}T_{2N}^{-1}VK$. Then
$T_{2N}^*=(\bar\mu/\mu)^{2N}V^{-1}T_{2N}V$ is obtained by using
$T_{2N}^TKT_{2N}=(\mu/\bar\mu)^{2N}K$.

If $(T_{2N})_{jj}$ is diagonal whenever $r_j>1$, then
$T_{2N}^*=(\bar\mu/\mu)^{2N}V^{-1}T_{2N}V$ implies that
$(\bar\mu/\mu)^NT_{2N}$ is a real diagonal matrix. In this case,
$\widetilde V$ is a real diagonal matrix satisfying
$\Omega\widetilde V\Omega^{-1}=\widetilde V$, $\widetilde
V^TK\widetilde V=K$. The theorem is proved.
\end{proof}

According to Theorem~\ref{thm:DT_sys}, to verify that $T(\lambda)$
is a Darboux matrix keeping all the symmetries in (\ref{eq:Vred}),
it is only necessary to verify that $(T_{2N})_{jj}$ is diagonal
whenever $r_j>1$. In the rest of this section, we will show this
fact and simplify the expression of the solutions .

\begin{lemma}\label{lemma:sum}
Let $N$ be a positive integer, $p$ be an integer, $a$ be a complex
number. Let $\omega=\exp(2\pi\I/N)$. For any integer $k$, let
$\{k\}$ be the remainder of $k$ divided by $N$. Then
\begin{equation}
   \sum_{i=1}^N\frac{\omega^{-p(i-1)}}{1-a\omega^{i-1}}=\frac{Na^{\{p\}}}{1-a^N}
   \label{eq:lemma_sum_1}
\end{equation}
if $a^N\ne 1$, and
\begin{equation}
   \sum_{i=1\atop i\ne k}^N\frac{\omega^{-p(i-1)}}{1-a\omega^{i-1}}
   =\Big(\frac{N-1}2-\{p\}\Big)a^{\{p\}} \label{eq:lemma_sum_2}
\end{equation}
if $a=\omega^{-k+1}$.
\end{lemma}

\begin{proof}
Suppose $|a|<1$. Since $\D\sum_{i=1}^N\omega^{k(i-1)}=N\delta_{k0}$,
we have
\begin{equation}
   \begin{array}{l}
   \D\sum_{i=1}^N\frac{\omega^{-p(i-1)}}{1-a\omega^{i-1}}
   =\sum_{k=0}^\infty\sum_{i=1}^N a^k\omega^{(k-p)(i-1)}
   =N\sum_{k=0\atop k-p\equiv 0\,\hbox{\scriptsize mod}\,N}^\infty a^k\\
   \D=N\sum_{j=0}^\infty a^{Nj+\{p\}}
   =\frac{Na^{\{p\}}}{1-a^N}.
   \end{array}
\end{equation}
Similar conclusion holds for $a$ with $|a|>1$. When $|a|=1$ but
$a^N\ne 1$, (\ref{eq:lemma_sum_1}) is derived by a limit. When
$a=\omega^{-k+1}$, (\ref{eq:lemma_sum_2}) is obtained similarly
without considering the term with $i=k$. The lemma is proved.
\end{proof}

Let $e_a$ $(a=1,\cdots,N)$ be $N$ dimensional vectors with the
entries $(e_a)_k=\omega^{-(a-1)(k-1)}$, then $e_a^*
e_b=N\delta_{ab}$. Write
$H=\left(\begin{array}{c}h_1\\\vdots\\h_N\end{array}\right)$ where
$h_j$ is an $r_j\times 1$ column vector. By (\ref{ABCD1}),
(\ref{ABCD2}) and Lemma~\ref{lemma:sum},
\begin{equation}
   \begin{array}{l}
   \D(Ae_a)_j=\sum_{k=1}^NA_{jk}(e_a)_k\\
   \D=\sum_{k,c=1}^N\frac{\omega^{(1-m)(j-1)}\omega^{-(k-j)(c-1)}\omega^{-(a-1)(k-1)}
   h_{m-c}^TK_{m-c,c}h_c}{(1+\omega^{k-j})\mu}\\
   \D=\sum_{c=1}^N\frac{N}{2\mu}(-1)^{\{a+c-2\}}h_{m-c}^TK_{m-c,c}h_{c}(e_{a+m-1})_j
   \end{array}
\end{equation}
when $N$ is odd, and
\begin{equation}
   \begin{array}{l}
   \D(Ae_a)_j=
   \sum_{k,c=1\atop k\ne j+\varepsilon n}^N\frac{\omega^{(1-m)(j-1)}
   \omega^{-(k-j)(c-1)}\omega^{-(a-1)(k-1)}
   h_{m-c}^TK_{m-c,c}h_c}{(1+\omega^{k-j})\mu}\\
   \D\qquad+\sum_{c=1}^N\omega^{(1-m)(j-1)}(-1)^{c-1}\omega^{-(a-1)(j+\varepsilon n-1)}
   h_{m-c}^{\prime\,T}K_{m-c,c}h_c\\
   \D=\sum_{c=1}^N(-1)^{a+c-1}\Big(\{a+c-2\}\mu^{-1}
   h_{m-c}^TK_{m-c,c}h_c
   -\sigma(m)h^{\prime\,T}_{m-c}K_{m-c,c}h_c\Big)\\
   \quad\cdot(e_{a+m-1})_j
   \end{array}
\end{equation}
when $N$ is even with
$\D\sigma(m)=\left\{\begin{array}{ll}1\quad&\hbox{if $m$ is
odd}\\0&\hbox{if $m$ is even}\end{array}\right.$. Here we use the
facts $(-1)^{\{k\}}=(-1)^k$ and $H^TK\Omega^nH=0$ for even $N$, and
$H^{\prime\,T}K\Omega^nH=0$ for even $N$ and even $m$.

Likewise, we can get all the expressions of $Ae_a$, $Be_a$, $Ce_a$
and $De_a$ for both odd and even $N$. They are
\begin{equation}
   \begin{array}{l}
   \D Ae_a=N\mu^{-1}\alpha_ae_{a+m-1},\quad
   Be_a=N\mu^{-1}\beta_ae_{a+1},\\
   Ce_a=N\bar\mu\gamma_ae_{a+1},\quad
   De_a=N\bar\mu\delta_ae_{a-m+3}
   \end{array}
\end{equation}
where
\begin{equation}
   \begin{array}{l}
   \D\alpha_a=\frac{1}{2}\sum_{c=1}^N(-1)^{\{a+c-2\}}h_{m-c}^TK_{m-c,c}h_{c},\\
   \D \beta_a=|\mu|^2\sum_{c=1}^N
   \frac{(-|\mu|^2)^{\{c-a-1\}}}{1+|\mu|^{2N}}h_c^* V_{cc}^{-1}h_c,\\
   \D \gamma_a=\sum_{c=1}^N
   \frac{(-|\mu|^2)^{\{a+c-2\}}}{1+|\mu|^{2N}}h_c^* V_{cc}^{-1}h_c,\\
   \D\delta_a=\frac 1{2}\sum_{c=1}^N(-1)^{\{a-c\}}h_{m-c}^* K_{m-c,c}\bar h_c
   \end{array}\label{eq:alphabetagammadeltaodd}
\end{equation}
for odd $N$, and
\begin{equation}
   \begin{array}{l}
   \D\alpha_a=\frac 1{N}\sum_{c=1}^N(-1)^{a+c-1}\Big(\{a+c-2\}
   h_{m-c}^TK_{m-c,c}h_c-\sigma(m)\mu h^{\prime\,T}_{m-c}K_{m-c,c}h_c\Big)\\
   \D \beta_a=|\mu|^2\sum_{c=1}^N
   \frac{(-|\mu|^2)^{\{c-a-1\}}}{1-|\mu|^{2N}}h_c^* V_{cc}^{-1}h_c,\\
   \D \gamma_a=\sum_{c=1}^N
   \frac{(-|\mu|^2)^{\{a+c-2\}}}{1-|\mu|^{2N}}h_c^* V_{cc}^{-1}h_c,\\
   \D\delta_a=\frac 1{N}\sum_{c=1}^N(-1)^{a+c-1}\Big(\{a-c\}
   h_{m-c}^* K_{m-c,c}\bar h_c+\sigma(m)\bar\mu h^{\prime\,*}_{m-c}K_{m-c,c}\bar
   h_c\Big)
   \end{array}\label{eq:alphabetagammadeltaeven}
\end{equation}
for even $N$.

Since $\Gamma^{-1}=\left(\begin{array}{cc}\check A&\check B\\\check
C&\check D\end{array}\right)$ where
\begin{equation}
   \begin{array}{l}
   \check A=(A-BD^{-1}C)^{-1},\quad
   \check B=-A^{-1}B(D-CA^{-1}B)^{-1},\\
   \check C=-D^{-1}C(A-BD^{-1}C)^{-1},\quad
   \check D=(D-CA^{-1}B)^{-1},
   \end{array}\label{eq:GammaInv}
\end{equation}
we have
\begin{equation}
   \begin{array}{l}
   \D \check Ae_a=\frac\mu N\frac{\delta_{a-1}}{\Delta_{a+1}}e_{a-m+1},\quad
   \check Be_a=-\frac 1{N\bar\mu}\frac{\beta_{a+m-3}}{\Delta_{a+m-1}}e_{a-1},\\
   \D \check Ce_a=-\frac\mu N\frac{\gamma_{a-m+1}}{\Delta_{a+1}}e_{a-1},\quad
   \check De_a=\frac 1{N\bar\mu}\frac{\alpha_{a-1}}{\Delta_{a+m-1}}e_{a+m-3}
   \end{array}\label{eq:OddGenGammaInvExpr}
\end{equation}
where
\begin{equation}
   \Delta_a=\alpha_{a-m}\delta_{a-2}-\beta_{a-2}\gamma_{a-m}.\label{eq:Delta1}
\end{equation}
That is,
\begin{equation}
   \begin{array}{l}
   \D\sum_{b=1}^N\check A_{ab}\omega^{-(j-1)(b-1)}
   =\frac\mu
   N\frac{\delta_{j-1}}{\Delta_{j+1}}\omega^{-(j-m)(a-1)},\\
   \D\sum_{b=1}^N\check B_{ab}\omega^{-(j-1)(b-1)}
   =-\frac 1{N\bar\mu}\frac{\beta_{j+m-3}}{\Delta_{j+m-1}}\omega^{-(j-2)(a-1)},\\
   \D\sum_{b=1}^N\check C_{ab}\omega^{-(j-1)(b-1)}
   =-\frac\mu
   N\frac{\gamma_{j-m+1}}{\Delta_{j+1}}\omega^{-(j-2)(a-1)},\\
   \D\sum_{b=1}^N\check D_{ab}\omega^{-(j-1)(b-1)}
   =\frac 1{N\bar\mu}\frac{\alpha_{j-1}}{\Delta_{j+m-1}}\omega^{-(j+m-4)(a-1)}.
   \end{array}\label{eq:checkAexpr}
\end{equation}

(\ref{eq:alphabetagammadeltaodd}),
(\ref{eq:alphabetagammadeltaeven}) and the identity
$\{-p-1\}+\{p\}=N-1$ lead to the following relations.

\begin{lemma}
The coefficients $\alpha_k,\beta_k,\gamma_k,\delta_k$ satisfy
\begin{equation}
   \begin{array}{l}
   \alpha_{3-m-k}=\alpha_k,\quad
   \bar\delta_k=\alpha_{k-m+2},\\
   \bar\beta_k=\beta_k,\quad\bar\gamma_k=\gamma_k,\quad
   \beta_k=|\mu|^2\gamma_{1-k}.
   \end{array}\label{eq:alphabetagammadeltarel}
\end{equation}
Moreover,
\begin{equation}
   \Delta_k=|\alpha_{k-m}|^2-\beta_{k-2}\gamma_{k-m} \label{eq:Delta}
\end{equation}
is real and satisfies
\begin{equation}
   \Delta_{3+m-k}=\Delta_k.
   \label{eq:Deltarel}
\end{equation}
\end{lemma}

From (\ref{eq:checkAexpr}), we have
\begin{equation}
   \begin{array}{l}
   \D\sum_{a,b=1}^N(\mu_b^{-1} H_a\check A_{ab}H_b^TK)_{jk}
   =\Delta_{2+m-j}^{-1}\bar\alpha_{2-j} h_jh_{m-j}^TK_{m-j,j}\delta_{jk},\\
   \D\sum_{a,b=1}^N(\mu_{N+b}^{-1} H_a\check B_{ab}H_{N+b}^TK)_{jk}
   =-\Delta_{2+m-j}^{-1}\beta_{m-j}h_jh_j^* V_{jj}^{-1}\delta_{jk},\\
   \D\sum_{a,b=1}^N(\mu_b^{-1} H_{N+a}\check C_{ab}H_b^TK)_{jk}\\
   \qquad=-\Delta_{2+m-j}^{-1}\gamma_{2-j}K_{j,m-j}V_{m-j,m-j}^{-1}
   \bar h_{m-j}h_{m-j}^TK_{m-j,j}\delta_{jk},\\
   \D\sum_{a,b=1}^N(\mu_{N+b}^{-1} H_{N+a}\check
   D_{ab}H_{N+b}^TK)_{jk}\\
   \qquad=\Delta_{2+m-j}^{-1}\alpha_{2-j}K_{j,m-j}
   V_{m-j,m-j}^{-1}\bar h_{m-j}h_j^* V_{jj}^{-1}\delta_{jk}
   \end{array}\label{eq:beforeF}
\end{equation}
for $a,b=1,\cdots,N$.

Denote $\D T_{2N}=(\mu/\bar\mu)^NF$ where $\D
F=I-\sum_{a,b=1}^{2N}\mu_b^{-1} H_a\check\Gamma_{ab}H_k^TK$.
(\ref{eq:beforeF}) leads to $F=(f_j\delta_{jk})_{1\le j,k\le N}$
where
\begin{equation}
   \begin{array}{l}
   \D f_j\!=\!I_{r_j}-\Delta_{2+m-j}^{-1}
   \Big(\bar\alpha_{2-j}h_jh_{m-j}^TK_{m-j,j}
   +\alpha_{2-j}K_{j,m-j}V_{m-j,m-j}^{-1}\bar h_{m-j}h_j^* V_{jj}^{-1}\\
   \D\qquad-\beta_{m-j}h_jh_j^* V_{jj}^{-1}
   -\gamma_{2-j}K_{j,m-j}V_{m-j,m-j}^{-1}
   \bar h_{m-j}h_{m-j}^TK_{m-j,j}\Big)
   \end{array}\label{eq:f}
\end{equation}
is an $r_j\times r_j$ matrix.

Computing $f_j$'s explicitly, we get the expressions of new
solutions of the two dimensional elliptic affine Toda equation
(\ref{eq:intc}) from known ones by Darboux transformation.

\begin{theorem}\label{thm}
Let $V(z,\zbar)$ be a real diagonal matrix solution of the two
dimensional elliptic affine Toda equation (\ref{eq:intc}),
satisfying $V^TKV=K$ and $\det V_{jj}=1$ for $j$ with $2j\equiv
m\mod N$. Let $\D
H=\left(\begin{array}{cc}h_1\\\vdots\\h_N\end{array}\right)$ be a
column solution of the Lax pair~(\ref{eq:LP}) with $\lambda=\mu$
$(\mu\ne 0,|\mu|\ne 1)$ where $h_j$ is an $r_j\times 1$ column
vector. When both $N$ and $m$ are even, $H$ should also satisfy the
constraint $H^TK\Omega^{N/2}H=0$. Let $\alpha_j$, $\beta_j$,
$\gamma_j$ be given by (\ref{eq:alphabetagammadeltaodd}) and
(\ref{eq:alphabetagammadeltaeven}) for odd and even $N$
respectively, and $\Delta_j$ be given by (\ref{eq:Delta}). When
$r_j=2$, denote $\D
V_{jj}=\left(\begin{array}{cc}v_j\\&v_j^{-1}\end{array}\right)$, $\D
h_j=\left(\begin{array}{c}h_{j1}\\h_{j2}\end{array}\right)$. Then
\begin{equation}
   f_j=\frac{\Delta_{1+m-j}}{\Delta_{2+m-j}}
\end{equation}
if $r_j=1$, and
\begin{equation}
   f_j=\left(\begin{array}{cc}
   \D\frac{|h_{j1}|^2v_j^{-1}-\gamma_{2-j}}{|h_{j2}|^2v_j-\gamma_{2-j}}&0\\
   0&\D\frac{|h_{j2}|^2v_j-\gamma_{2-j}}{|h_{j1}|^2v_j^{-1}-\gamma_{2-j}}
   \end{array}\right)
\end{equation}
if $r_j=2$. Moreover, $\det f_j=1$ for $j$ with $2j\equiv m\mod N$
no matter $r_j=1$ or $2$. Therefore, $\widetilde
V=(V_{jj}f_j\delta_{jk})_{1\le j,k\le N}$ is a new real diagonal
solution of (\ref{eq:intc}) provided that $f_j$ is positive
definite, and satisfies $\widetilde V^TK\widetilde V=K$ and
$\det\widetilde V_{jj}=1$ for $j$ with $2j\equiv m\mod N$.
\end{theorem}

\begin{proof}
First suppose $r_j=2$. As demanded at the beginning of
Section~\ref{sect:LP}, $r_j=2$ holds only when $2j\equiv m\mod N$.

When $N=2n$ is even, $2j\equiv m\mod N$ happens only when $m$ is
even, and $j=m/2$ or $j=n+m/2$ in this case. For this $j$, the
constraint $H^TK\Omega^nH=0$ implies
\begin{equation}
   h_j^TK_{jj}h_j+(-1)^{n}h_{n+j}^TK_{n+j,n+j}h_{n+j}
   +2\sum_{k=1}^{n-1}(-1)^kh_{j-k}^TK_{j-k,j+k}h_{j+k}=0.
\end{equation}
With this constraint,
\begin{equation}
   \begin{array}{l}
   \D\alpha_{2-j}=\frac 1{2n}\sum_{c=1}^{2n}(-1)^{j+c+1}\{c-j\}
   h_{2j-c}^TK_{2j-c,c}h_{c}\\
   \D=\frac 1{2}(-1)^{n+1}h_{n+j}^TK_{n+j,n+j}h_{n+j}
   +\sum_{k=1}^{n-1}(-1)^{k+1}h_{j-k}^TK_{j-k,j+k}h_{j+k}\\
   \D=\frac 12h_j^TK_{jj}h_j=h_{j1}h_{j2},
   \end{array}
\end{equation}
\begin{equation}
   \begin{array}{l}
   \D\beta_j+\gamma_{2-j}=|\mu|^2\sum_{c=1}^{2n}\frac{(-|\mu|^2)^{\{c-j-1\}}}{1-|\mu|^{4n}}
   h_c^*V_{cc}^{-1}h_c+\sum_{c=1}^{2n}\frac{(-|\mu|^2)^{\{c-j\}}}{1-|\mu|^{4n}}
   h_c^* V_{cc}^{-1}h_c\\
   \D=h_j^* V_{jj}^{-1} h_j=|h_{j1}|^2v_j^{-1}+|h_{j2}|^2v_j.
   \end{array}\label{eq:betagamma1}
\end{equation}
Hence if $r_j=2$,
\begin{equation}
   \begin{array}{l}
   \D f_j=I_2-\Delta_{2+j}^{-1}\Big(\bar\alpha_{2-j}h_jh_j^TK_{jj}
   +\alpha_{2-j}K_{jj}V_{jj}^{-1}\bar h_jh_j^* V_{jj}^{-1}\\
   \qquad -\beta_jh_jh_j^* V_{jj}^{-1}
   -\gamma_{2-j}K_{jj}V_{jj}^{-1}\bar h_jh_j^TK_{jj}\Big)\\[6pt]
   \D=
   \left(\begin{array}{cc}
   \D\frac{|h_{j1}|^2v_j^{-1}-\gamma_{2-j}}{|h_{j2}|^2v_j-\gamma_{2-j}}&0\\
   0&\D\frac{|h_{j2}|^2v_j-\gamma_{2-j}}{|h_{j1}|^2v_j^{-1}-\gamma_{2-j}}
   \end{array}\right)
   \end{array}\label{eq:explicit_f1}
\end{equation}
is a diagonal matrix with $\det f_j=1$, since
\begin{equation}
   \Delta_{2+j}=(|h_{j1}|^2v_j^{-1}-\gamma_{2-j})(|h_{j2}|^2v_j-\gamma_{2-j}).
\end{equation}

When $N$ is odd, there is no constraint on $H$. From
(\ref{eq:alphabetagammadeltaodd}), the equalities
$\D\alpha_{2-j}=h_{j1}h_{j2}$ and
$\D\beta_j+\gamma_{2-j}=|h_{j1}|^2v_j^{-1}+|h_{j2}|^2v_j$ still hold
for $j$ with $2j\equiv m\mod N$. Hence (\ref{eq:explicit_f1}) is
still true when $N$ is odd and $r_j=2$.

Now we consider the case $r_j=1$. Denote
$\sigma_j=V_{jj}^{-1}|h_j|^2$, $\tau_j=h_{m-j}h_j$, then
\begin{equation}
   f_j=1-\Delta_{2+m-j}^{-1}(\bar\alpha_{2-j}\tau_j+\alpha_{2-j}\bar\tau_j
   -\beta_{m-j}\sigma_j-r_{2-j}\sigma_{m-j})
\end{equation}
since $V_{m-j,m-j}V_{jj}=1$. By direct calculation, we have
\begin{equation}
   \begin{array}{l}
   \beta_{m-j-1}+|\mu|^2\beta_{m-j}=|\mu|^2\sigma_{m-j},\\
   \gamma_{2-j}+|\mu|^2\gamma_{1-j}=\sigma_j,\\
   \alpha_{2-j}+\alpha_{1-j}=\tau_j.\label{eq:alphaetcneigh}
   \end{array}
\end{equation}
Here the relation $H^TK\Omega^nH=0$ is used for even $N$.
(\ref{eq:alphaetcneigh}) leads to
\begin{equation}
   \begin{array}{l}
   \Delta_{2+m-j}-\Delta_{1+m-j}=|\alpha_{2-j}|^2-|\alpha_{1-j}|^2
   -\beta_{m-j}\gamma_{2-j}+\beta_{m-j-1}\gamma_{1-j}\\
   =\bar\alpha_{2-j}\tau_j+\alpha_{2-j}\bar\tau_j
   -\beta_{m-j}\sigma_j-r_{2-j}\sigma_{m-j}+\sigma_j\sigma_{m-j}-|\tau_j|^2\\
   =\bar\alpha_{2-j}\tau_j+\alpha_{2-j}\bar\tau_j
   -\beta_{m-j}\sigma_j-r_{2-j}\sigma_{m-j}
   \end{array}
\end{equation}
by using $V_{m-j,m-j}V_{jj}=1$. Hence
\begin{equation}
   f_j=\frac{\Delta_{1+m-j}}{\Delta_{2+m-j}}.
\end{equation}
If $j$ satisfies $2j\equiv m\mod N$, then (\ref{eq:Deltarel})
implies $f_j=1$. The theorem is proved.
\end{proof}

According to different choices of $N$, $K$ and $r_1,\cdots,r_N$, the
local solutions of various two dimensional elliptic affine Toda
equations can be obtained from Theorem~\ref{thm}, which will be
shown in the next section.

Although in Theorem~\ref{thm}, only Darboux transformation with one
spectral parameter $\mu$ is presented, Darboux transformation of
higher degree can be obtained by the composition of this kind of
Darboux transformations.

\section{Darboux transformations for two dimensional elliptic
$A_l^{(1)}$, $A_{2l}^{(2)}$, $A_{2l-1}^{(2)}$, $B_l^{(1)}$,
$C_l^{(1)}$, $D_l^{(1)}$ and $D_{l+1}^{(2)}$ Toda
equations}\label{sect:Egs}

In this section we present the explicit expressions of new solutions
given by Darboux transformations for various two dimensional
elliptic affine Toda equations.

The two dimensional elliptic $A_l^{(1)}$ Toda equation is not
contained in the system discussed in Section~\ref{sect:LP} and
\ref{sect:DT}, since it does not possess the complex orthogonal
symmetry. However, a simpler Darboux transformation can be
constructed due to the lack of this symmetry.

For the rest six series of equations the results presented below are
the consequences of Theorem~\ref{thm} directly.

\subsection{Two dimensional elliptic $A_{l}^{(1)}$ Toda equation $(l\ge 1)$}\label{sect:A1}\

In this case, $M=N=l+1$, $r_1=\cdots=r_{l+1}=1$;
$\D\omega=\exp\Big(\frac{2\pi\I}{l+1}\Big)$;
$\Omega=(\Omega_{jk})_{1\le j,k\le l+1}$ with
$\Omega_{jk}=\omega^{-j+1}\delta_{jk}$; $J=(J_{jk})_{1\le j,k\le
l+1}$ with $J_{jk}=\delta_{j+1,k}$; $V=(V_{jk})_{1\le j,k\le l+1}$
with $V_{jk}=\E{u_j}\delta_{jk}$ and $\D\sum_{j=1}^{l+1} u_j=0$.
There is no matrix $K$ defining complex orthogonal symmetry.

The equation is
\begin{equation}
   \triangle u_j=\E{u_j-u_{j-1}}-\E{u_{j+1}-u_j}\quad
   (j=1,\cdots,l+1) \label{eq:eqA1}
\end{equation}
with $u_0=u_{l+1}$, $u_{l+2}=u_1$. By setting
$w_j=-(u_1+\cdots+u_j)$ $(j=1,\cdots,l)$, it becomes
\begin{equation}
   \begin{array}{l}
   \triangle w_1=\E{2w_1-w_2}-\E{-w_1-w_l},\\
   \triangle
   w_j=\E{-w_{j-1}+2w_j-w_{j+1}}-\E{-w_1-w_l}\quad(j=2,\cdots,l-1),\\
   \triangle w_l=\E{-w_{l-1}+2w_l}-\E{-w_1-w_l}.
   \end{array}
\end{equation}

Similar to (ii) and (iii) of Lemma~\ref{lemma:sys}, we know that if
$\Phi$ is a solution of the Lax pair with $\lambda=\lambda_0$,
$V^{-1}\bar\Phi$ is a solution of the adjoint Lax pair with
$\lambda=-\bar\lambda_0^{-1}$. (This symmetry can also be changed to
a unitary symmetry by a gauge transformation
\cite{bib:McIntosh,bib:Terng}.) Owing to this fact, the Darboux
transformation can be constructed as follows.

Let $\mu$ be a non-zero complex number with $|\mu|\ne 1$. Let $H$ be
a column solution of the Lax pair (\ref{eq:LP}). Let
$\mu_j=\omega^{j-1}\mu$, $H_j=\Omega^{j-1}H$ $(j=1,\cdots,N)$ with
$N=l+1$. Let $\D\Gamma_{jk}=\frac{H_j^* V^{-1}
H_k}{1+\bar\mu_j\mu_k}$, then $\Gamma_{j+1,k+1}=\Gamma_{jk}$,
$\Gamma_{jk}^*=\Gamma_{kj}$, $\Gamma_{jk,z}=\bar\mu_j^{-1} H_j^*
V^{-1} JH_k$, $\Gamma_{jk,\zbar}=\mu_k^{-1} H_j^*J^TV^{-1}H_k$.

Let
\begin{equation}
   T(\lambda)=\prod_{i=1}^N(\lambda+\bar\mu_i^{-1})\Big(I-\sum_{j,k=1}^N
   \frac{H_j\check\Gamma_{jk}H_k^* V^{-1}}{\bar\mu_k\lambda+1}\Big),
\end{equation}
where $\check\Gamma=\Gamma^{-1}$, then
\begin{equation}
   T(\lambda)^{-1}=\prod_{i=1}^N(\lambda+\bar\mu_i^{-1})^{-1}\Big(I+\sum_{j,k=1}^N
   \frac{H_j\check\Gamma_{jk}H_k^* V^{-1}}{\bar\mu_k(\lambda-\mu_j)}\Big).
\end{equation}

Similar to Lemma~\ref{lemma:DTPQV}, we know that $\widetilde
V=\rho^{-1} T_NV$ satisfies (\ref{eq:intc}) for any non-zero complex
number $\rho$. What is more, similar to Theorem~\ref{thm:DT_sys},
\begin{equation}
   \begin{array}{l}
   \D\Omega T(\lambda)\Omega^{-1}=T(\omega\lambda),\\
   \D\lambda^{N}VT(-\bar\lambda^{-1})^* V^{-1} T_N^{-1}
   =\prod_{i=1}^N(\mu_i^{-1}\lambda-1)(\bar\mu_i\lambda+1)T(\lambda)^{-1}
   \end{array}
\end{equation}
hold by simpler computation which is omitted here. Especially, $T_N$
satisfies $\Omega T_N\Omega^{-1}=T_N$, $T_N^*=(\bar\mu/\mu)^NV^{-1}
T_NV=(\bar\mu/\mu)^NT_N$ since all $r_j=1$. Hence
$(-1)^N\mu^{-N}T_N$ is a real diagonal matrix.

Written in terms of the basis $\{e_\alpha\}$, we have
\begin{equation}
   \Gamma e_a=N\alpha_ae_a,\quad \Gamma^{-1} e_a=\frac 1N\alpha_a^{-1} e_a
\end{equation}
where
\begin{equation}
   \alpha_a=\sum_{c=1}^N\frac{(-|\mu|^2)^{\{a+c-2\}}}{1-(-|\mu|^2)^N}|h_c|^2V_{cc}^{-1}.
   \label{eq:A1alpha}
\end{equation}

Then $T_N=((-1)^N\bar\mu^{-N}|\mu|^2f_j\delta_{jk})_{1\le j,k\le N}$
where
\begin{equation}
   f_j=|\mu|^{-2}(\alpha_{2-j}^{-1}|h_j|^2V_{jj}^{-1}-1)=\frac{\alpha_{1-j}}{\alpha_{2-j}}
\end{equation}
by using (\ref{eq:A1alpha}). Hence we get the following explicit
result for the new solution $\widetilde
V=(f_jV_{jj}\delta_{jk})_{1\le j,k\le l+1}$.

Let $(u_1,\cdots,u_{l+1})$ be a solution of (\ref{eq:eqA1}),
$(h_1,\cdots,h_{l+1})^T$ be a column solution of the corresponding
Lax pair where $h_1,\cdots,h_{l+1}$ are scalar functions. Then the
new solution of (\ref{eq:eqA1}) is $\D\widetilde
u_j=u_j+\ln\frac{g_{j+1}}{g_j}$ $(j=1,\cdots,l+1)$ with
\begin{equation}
   g_j=\sum_{c=1}^{l+1}(-|\mu|^2)^{\{c-j\}}|h_c|^2\E{-u_c}.
\end{equation}

\subsection{Two dimensional elliptic $A_{2l}^{(2)}$ Toda equation $(l\ge 1)$}\

In this case, $M=N=2l+1$, $r_1,\cdots,r_{2l+1}=1$, $m=2$;
$\D\omega=\exp\Big(\frac{2\pi\I}{2l+1}\Big)$;
$\D\Omega=(\Omega_{jk})_{1\le j,k\le 2l+1}$ with
$\Omega_{jj}=\omega^{-j+1}\delta_{jk}$; $\D K=(K_{jk})_{1\le j,k\le
2l+1}$ with $K_{jk}=\delta_{j+k,2}$; $J=(J_{jk})_{1\le j,k\le 2l+1}$
with $J_{jk}=\delta_{j+1,k}$; $V=(V_{jk})_{1\le j,k\le 2l+1}$ with
non-zero entries $V_{11}=1$, $V_{j+1,j+1}=\E{u_j}$ and
$V_{2l+2-j,2l+2-j}=\E{-u_j}$ for $j=1,\cdots,l$. The equation is
\begin{equation}
   \begin{array}{l}
   \D\triangle u_1=\E{u_1}-\E{u_2-u_1},\\
   \D\triangle u_j=\E{u_j-u_{j-1}}-\E{u_{j+1}-u_j}\quad
   (j=2,\cdots,l-1),\\
   \D\triangle u_l=\E{u_l-u_{l-1}}-\E{-2u_l}
   \end{array}\label{eq:eqA2even}
\end{equation}
when $l\ge 2$ and
\begin{equation}
   \begin{array}{l}
   \D\triangle u_1=\E{u_1}-\E{-2u_1}.\label{eq:eqA2even1}
   \end{array}
\end{equation}
when $l=1$. By setting $w_j=-(u_1+\cdots+u_j)$ $(j=1,\cdots,l)$,
they become
\begin{equation}
   \begin{array}{l}
   \D\triangle w_1=\E{2w_1-w_2}-\E{-w_1},\\
   \D\triangle w_j=\E{-w_{j-1}+2w_j-w_{j+1}}-\E{-w_1}\quad(j=2,\cdots,l-1),\\
   \D\triangle w_l=\E{-2w_{l-1}+2w_l}-\E{-w_1}
   \end{array}
\end{equation}
and
\begin{equation}
   \begin{array}{l}
   \D\triangle w_1=\E{2w_1}-\E{-w_1}
   \end{array}
\end{equation}
respectively.

Let $(u_1,\cdots,u_l)$ be a solution of (\ref{eq:eqA2even}) or
(\ref{eq:eqA2even1}), $(h_1,\cdots,h_{2l+1})^T$ be a column solution
of the corresponding Lax pair where $h_1,\cdots,h_{2l+1}$ are scalar
functions. By Theorem~\ref{thm}, the new solution of
(\ref{eq:eqA2even}) or (\ref{eq:eqA2even1}) is $\D\widetilde
u_j=u_j+\ln\frac{g_{j+1}}{g_j}$ $(j=1,\cdots,l)$ where
\begin{equation}
   \begin{array}{rl}
   g_j=&\D\frac14\Big|\sum_{c=1}^{2l+1}(-1)^{\{c-j\}}h_{2-c}h_c\Big|^2
   -|\mu|^2\gamma_{j-1}\gamma_{2-j}
   \end{array}
\end{equation}
and
\begin{equation}
   \begin{array}{rl}
   \gamma_j=&\D\frac 1{1+|\mu|^{4l+2}}\Bigg((-|\mu|^2)^{\{j-1\}}|h_1|^2+
   \sum_{c=1}^{l}
   (-|\mu|^2)^{\{j+c-1\}}|h_{1+c}|^2\E{-u_c}\\
   &\D+\sum_{c=1}^{l}
   (-|\mu|^2)^{\{j-c-1\}}|h_{1-c}|^2\E{u_c}\Bigg).
   \end{array}
\end{equation}

\subsection{Two dimensional elliptic $C_l^{(1)}$ Toda equation $(l\ge 2)$}\

In this case, $M=N=2l$, $r_1,\cdots,r_{2l}=1$, $m=1$;
$\D\omega=\exp\Big(\frac{2\pi\I}{2l}\Big)$;
$\D\Omega=(\Omega_{jk})_{1\le j,k\le 2l}$ with
$\Omega_{jk}=\omega^{-j+1}\delta_{jk}$; $\D K=(K_{jk})_{1\le j,k\le
2l}$ with $K_{jk}=\delta_{j+k,1}$; $J=(J_{jk})_{1\le j,k\le 2l}$
with $J_{jk}=\delta_{j+1,k}$; $\D V=(V_{jk})_{1\le j,k\le 2l}$ with
non-zero entries $V_{jj}=\E{u_j}$ and $V_{2l+1-j,2l+1-j}=\E{-u_j}$
for $j=1,\cdots,l$. The equation is
\begin{equation}
   \begin{array}{l}
   \D\triangle u_1=\E{2u_1}-\E{u_2-u_1},\\
   \D\triangle u_j=\E{u_j-u_{j-1}}-\E{u_{j+1}-u_j}\quad
   (j=2,\cdots,l-1),\\
   \D\triangle u_l=\E{u_{l}-u_{l-1}}-\E{-2u_{l}}.
   \end{array}\label{eq:eqC1}
\end{equation}
By setting $w_j=-(u_1+\cdots+u_j)$ $(j=1,\cdots,l)$, it becomes
\begin{equation}
   \begin{array}{l}
   \D\triangle w_1=\E{2w_1-w_2}-\E{-2w_1},\\
   \D\triangle w_j=\E{-w_{j-1}+2w_j-w_{j+1}}-\E{-2w_1}\quad(j=2,\cdots,l-1),\\
   \D\triangle w_l=\E{-2w_{l-1}+2w_l}-\E{-2w_1}.
   \end{array}
\end{equation}

Let $(u_1,\cdots,u_l)$ be a solution of (\ref{eq:eqC1}),
$(h_1,\cdots,h_{2l})^T$ be a column solution of the corresponding
Lax pair where $h_1,\cdots,h_{2l}$ are scalar functions. By
Theorem~\ref{thm}, the new solution of (\ref{eq:eqC1}) is
$\D\widetilde u_j=u_j+\ln\frac{g_{j+1}}{g_j}$ $(j=1,\cdots,l)$ where
\begin{equation}
   \begin{array}{rl}
   g_j=&\D\frac1{4l^2}\Big|\sum_{c=1}^{2l}(-1)^{c}\big(
   \{c-j\}h_{1-c}h_c-\mu h_{1-c}^{\prime}h_c\big)\Big|^2
   -|\mu|^2\gamma_j\gamma_{2-j}
   \end{array}
\end{equation}
and
\begin{equation}
   \begin{array}{rl}
   \gamma_j=&\D\frac 1{1-|\mu|^{4l}}\Bigg(\sum_{c=1}^{l}
   (-|\mu|^2)^{\{j+c-2\}}|h_{c}|^2\E{-u_c}
   +\sum_{c=1}^{l}
   (-|\mu|^2)^{\{j-c-1\}}|h_{1-c}|^2\E{u_c}\Bigg).
   \end{array}
\end{equation}

\subsection{Two dimensional elliptic $D_{l+1}^{(2)}$ Toda equation $(l\ge 2)$}\

In this case, $M=N=2l+2$, $r_1,\cdots,r_{2l+2}=1$,
$m=2$; $\D\omega=\exp\Big(\frac{2\pi\I}{2l+2}\Big)$;
$\D\Omega=(\Omega_{jk})_{1\le j,k\le 2l+2}$ with
$\Omega_{jk}=\omega^{-j+1}\delta_{jk}$; $\D K=(K_{jk})_{1\le j,k\le
2l+2}$ with $K_{jk}=\delta_{j+k,2}$; $J=(J_{jk})_{1\le j,k\le 2l+2}$
with $J_{jk}=\delta_{j+1,k}$; $V=(V_{jk})_{1\le j,k\le 2l+2}$ with
non-zero entries $V_{11}=1$, $V_{l+2,l+2}=1$, $V_{j+1,j+1}=\E{u_j}$
and $V_{2l+3-j,2l+3-j}=\E{-u_j}$ for $j=1,\cdots,l$. The equation is
\begin{equation}
   \begin{array}{l}
   \D\triangle u_1=\E{u_1}-\E{u_2-u_1},\\
   \D\triangle u_j=\E{u_j-u_{j-1}}-\E{u_{j+1}-u_j}\quad
   (j=2,\cdots,l-1),\\
   \D\triangle u_l=\E{u_{l}-u_{l-1}}-\E{-u_{l}}.
   \end{array}\label{eq:eqD2}
\end{equation}
By setting $w_j=-(u_1+\cdots+u_j)$ $(j=1,\cdots,l-1)$ and
$w_l=-\frac 12(u_1+\cdots+u_l)$, it becomes
\begin{equation}
   \begin{array}{l}
   \D\triangle w_1=\E{2w_1-w_2}-\E{-w_1},\\
   \D\triangle w_j=\E{-w_{j-1}+2w_j-w_{j+1}}-\E{-w_1}\quad(j=2,\cdots,l-2),\\
   \D\triangle
   w_{l-1}=\E{-w_{l-2}+2w_{l-1}-2w_l}-\E{-w_1},\quad
   \D\triangle w_l=\frac 12\E{-w_{l-1}+2w_l}-\frac
   12\E{-w_1}.
   \end{array}
\end{equation}

Let $(u_1,\cdots,u_l)$ be a solution of (\ref{eq:eqD2}),
$(h_1,\cdots,h_{2l+2})^T$ be a column solution of the corresponding
Lax pair satisfying
\begin{equation}
   h_1^2+(-1)^{l+1}h_{l+2}^2+2\sum_{c=2}^{l+1}(-1)^{c-1}h_ch_{2-c}=0
\end{equation}
where $h_1,\cdots,h_{2l+2}$ are scalar functions. By
Theorem~\ref{thm}, the new solution of (\ref{eq:eqD2}) is
$\D\widetilde u_j=u_j+\ln\frac{g_{j+1}}{g_j}$ $(j=1,\cdots,l)$ where
\begin{equation}
   \begin{array}{rl}
   g_j=&\D\frac1{(2l+2)^2}\Big|\sum_{c=1}^{2l+2}(-1)^{c}
   \{c-j\}h_{2-c}h_c\Big|^2-|\mu|^2\gamma_{j-1}\gamma_{2-j}
   \end{array}
\end{equation}
and
\begin{equation}
   \begin{array}{rl}
   \gamma_j=&\D\frac 1{1-|\mu|^{4l+4}}\Bigg((-|\mu|^2)^{\{j-1\}}|h_1|^2
   +(-|\mu|^2)^{\{j+l\}}|h_{l+2}|^2\\
   &\D+\sum_{c=1}^{l}
   (-|\mu|^2)^{\{j+c-1\}}|h_{1+c}|^2\E{-u_c}
   +\sum_{c=1}^{l}
   (-|\mu|^2)^{\{j-c-1\}}|h_{1-c}|^2\E{u_c}\Bigg).
   \end{array}
\end{equation}

\subsection{Two dimensional elliptic $A_{2l-1}^{(2)}$ Toda equation $(l\ge 3)$}\

In this case, $N=2l-1$, $M=N+1$, $r_1=2$, $r_2=\cdots=r_{2l-1}=1$,
$m=2$; $\D\omega=\exp\Big(\frac{2\pi\I}{2l-1}\Big)$;
$\Omega=(\Omega_{jk})_{1\le j,k\le 2l-1}$ with non-zero entries
$\D\Omega_{11}=\left(\begin{array}{cc}1\\&1\end{array}\right)$ and
$\Omega_{jj}=\omega^{-j+1}$ for $j=2,\cdots,2l-1$; $K=(K_{jk})_{1\le
j,k\le 2l-1}$ with non-zero entries $\D
K_{11}=\left(\begin{array}{cc}&1\\1\end{array}\right)$ and
$K_{j,2l+1-j}=1$ for $j=2,\cdots,2l-1$; $J=(J_{jk})_{1\le j,k\le
2l-1}$ with non-zero entries $\D
J_{12}=\left(\begin{array}{c}1\\1\end{array}\right)$, $\D
J_{2l-1,1}=\left(\begin{array}{cc}1&1\end{array}\right)$ and $\D
J_{j,j+1}=1$ for $j=2,\cdots,2l-2$; $V=(V_{jk})_{1\le j,k\le 2l-1}$
with non-zero entries $\D
V_{11}=\left(\begin{array}{cc}\E{u_1}\\&\E{-u_1}\end{array}\right)$,
$V_{jj}=\E{u_j}$ and $V_{2l+1-j,2l+1-j}=\E{-u_j}$ for
$j=2,\cdots,l$. The equation is
\begin{equation}
   \begin{array}{l}
   \D\triangle u_1=\E{u_2+u_1}-\E{u_2-u_1},\quad
   \D\triangle u_2=\E{u_2+u_1}+\E{u_2-u_1}-\E{u_3-u_2},\\
   \D\triangle u_j=\E{u_j-u_{j-1}}-\E{u_{j+1}-u_j}\quad(j=3,\cdots,l-1),\\
   \D\triangle u_l=\E{u_l-u_{l-1}}-\E{-2u_l}.
   \end{array}\label{eq:eqA2odd}
\end{equation}
By setting $w_j=-(u_1+u_2+\cdots+u_j)$ $(j=1,\cdots,l)$, it becomes
\begin{equation}
   \begin{array}{l}
   \D\triangle w_1=\E{2w_1-w_2}-\E{-w_2},\\
   \D\triangle w_j=\E{-w_{j-1}+2w_j-w_{j+1}}-2\E{-w_2}\quad(j=2,\cdots,l-1),\\
   \D\triangle w_l=\E{-2w_{l-1}+2w_l}-2\E{-w_2}.
   \end{array}
\end{equation}

Let $(u_1,\cdots,u_l)$ be a solution of (\ref{eq:eqA2odd}),
$\D\left(\begin{array}{c}h_1\\\vdots\\h_{2l-1}\end{array}\right)$ be
a column solution of the corresponding Lax pair where $\D
h_1=\left(\begin{array}{c}h_{11}\\h_{12}\end{array}\right)$ and
$h_{11},h_{12},h_2,\cdots,h_{2l-1}$ are scalar functions. By
Theorem~\ref{thm}, the new solution of (\ref{eq:eqA2odd}) is
\begin{equation}
   \begin{array}{l}
   \D \widetilde u_1=u_1+\ln\frac{\D|h_{11}|^2\E{-u_1}-\gamma_1}
   {\D|h_{12}|^2\E{u_1}-\gamma_1},\\
   \D \widetilde u_j=u_j+\ln\frac{g_{j+1}}{g_j}\quad
   \D(j=2,\cdots,l)
   \end{array}
\end{equation}
where
\begin{equation}
   \begin{array}{rl}
   g_j=&\D\frac14\Big|2h_{11}h_{12}
   +\sum_{c=2}^{2l-1}(-1)^{\{c-j\}-\{1-j\}}
   h_{2-c}h_c\Big|^2
   -|\mu|^2\gamma_{j-1}\gamma_{2-j}
   \end{array}
\end{equation}
and
\begin{equation}
   \begin{array}{rl}
   \gamma_j=&\D\frac 1{1+|\mu|^{4l-2}}\Bigg((-|\mu|^2)^{\{j-1\}}
   (|h_{11}|^2\E{-u_1}+|h_{12}|^2\E{u_1})\\
   &\D+\sum_{c=2}^{l}
   (-|\mu|^2)^{\{j+c-2\}}|h_c|^2\E{-u_c}
   +\sum_{c=2}^{l}
   (-|\mu|^2)^{\{j-c\}}|h_{2-c}|^2\E{u_c}\Bigg).
   \end{array}
\end{equation}

\subsection{Two dimensional elliptic $B_l^{(1)}$ Toda equation $(l\ge 3)$} \

In this case, $N=2l$, $M=N+1$, $r_1=2$, $r_2=\cdots=r_{2l}=1$,
$m=2$; $\D\omega=\exp\Big(\frac{2\pi\I}{2l}\Big)$;
$\D\Omega=(\Omega_{jk})_{1\le j,k\le 2l}$ with non-zero entries
$\D\Omega_{11}=\left(\begin{array}{cc}1\\&1\end{array}\right)$ and
$\Omega_{jj}=\omega^{-j+1}$ for $j=2,\cdots,2l$; $K=(K_{jk})_{1\le
j,k\le 2l}$ with non-zero entries $\D
K_{11}=\left(\begin{array}{cc}&1\\1\end{array}\right)$ and
$K_{j,2l+2-j}=1$ for $j=2,\cdots,2l$; $J=(J_{jk})_{1\le j,k\le 2l}$
with non-zero entries
$J_{12}=\left(\begin{array}{c}1\\1\end{array}\right)$,
$J_{2l,1}=\left(\begin{array}{cc}1&1\end{array}\right)$ and $\D
J_{j,j+1}=1$ for $j=2,\cdots,2l-1$; $V=(V_{jk})_{1\le j,k\le 2l}$
with non-zero entries
$V_{11}=\left(\begin{array}{cc}\E{u_1}\\&\E{-u_1}\end{array}\right)$,
$V_{l+1,l+1}=1$, $V_{jj}=\E{u_j}$ and $V_{2l+2-j,2l+2-j}=\E{-u_j}$
for $j=2,\cdots,l$. The equation is
\begin{equation}
   \begin{array}{l}
   \D\triangle u_1=\E{u_2+u_1}-\E{u_2-u_1},\quad
   \D\triangle u_2=\E{u_2+u_1}+\E{u_2-u_1}-\E{u_3-u_2},\\
   \D\triangle u_j=\E{u_j-u_{j-1}}-\E{u_{j+1}-u_j}\quad (j=3,\cdots,l-1),\\
   \D\triangle u_l=\E{u_l-u_{l-1}}-\E{-u_l}.
   \end{array}\label{eq:eqB1}
\end{equation}
By setting $w_j=-(u_1+\cdots+u_j)$ $(j=1,\cdots,l-1)$, $w_l=\frac
12(u_1+\cdots+u_l)$, it becomes
\begin{equation}
   \begin{array}{l}
   \D\triangle w_1=\E{2w_1-w_2}-\E{-w_2},\\
   \D\triangle w_j=\E{-w_{j-1}+2w_j-w_{j+1}}-2\E{-w_2}\quad(j=2,\cdots,l-2),\\
   \D\triangle
   w_{l-1}=\E{-w_{l-2}+2w_{l-1}-2w_l}-2\E{-w_2},\quad
   \D\triangle w_l=\frac 12\E{-w_{l-1}+2w_l}-\E{-w_2}.\\
   \end{array}
\end{equation}

Let $(u_1,\cdots,u_l)$ be a solution of (\ref{eq:eqB1}),
$\D\left(\begin{array}{c}h_1\\\vdots\\h_{2l}\end{array}\right)$ be a
column solution of the corresponding Lax pair satisfying
\begin{equation}
   2h_{11}h_{12}+(-1)^{l}h_{l+1}^2+2\sum_{c=2}^l(-1)^{c-1}h_ch_{2-c}=0
\end{equation}
where $\D
h_1=\left(\begin{array}{c}h_{11}\\h_{12}\end{array}\right)$ and
$h_{11},h_{12},h_2,\cdots,h_{2l}$ are scalar functions. By
Theorem~\ref{thm}, the new solution of (\ref{eq:eqB1}) is
\begin{equation}
   \begin{array}{l}
   \D \widetilde u_1=u_1+\ln\frac{\D|h_{11}|^2\E{-u_1}-\gamma_1}
   {\D|h_{12}|^2\E{u_1}-\gamma_1},\\
   \D \widetilde u_j=u_j+\ln\frac{g_{j+1}}{g_j}\quad
   \D(j=2,\cdots,l)
   \end{array}
\end{equation}
where
\begin{equation}
   g_j=\frac1{4l^2}\Big|2\{1-j\}h_{11}h_{12}
   +\sum_{c=2}^{2l}(-1)^{c+1}
   \{c-j\}h_{2-c}h_c\Big|^2
   -|\mu|^2\gamma_{j-1}\gamma_{2-j}
\end{equation}
and
\begin{equation}
   \begin{array}{rl}
   \gamma_j=&\D\frac 1{1-|\mu|^{4l}}
   \Bigg((-|\mu|^2)^{\{j-1\}}(|h_{11}|^2\E{-u_1}+|h_{12}|^2\E{u_1})
   +(-|\mu|^2)^{\{j+l-1\}}|h_{l+1}|^2\\
   &\D+\sum_{c=2}^{l}
   (-|\mu|^2)^{\{j+c-2\}}|h_c|^2\E{-u_c}
   +\sum_{c=2}^{l}
   (-|\mu|^2)^{\{j-c\}}|h_{2-c}|^2\E{u_c}\Bigg).
   \end{array}
\end{equation}

\subsection{Two dimensional elliptic $D_l^{(1)}$ Toda equation $(l\ge 4)$}\

In this case, $N=2l-2$, $M=N+2$, $r_1=r_{l+1}=2$,
$r_2=\cdots=r_l=r_{l+2}=\cdots=r_{2l}=1$, $m=2$;
$\D\omega=\exp\Big(\frac{2\pi\I}{2l-2}\Big)$;
$\Omega=(\Omega_{jk})_{1\le j,k\le 2l-2}$ with non-zero entries
$\D\Omega_{11}=\Omega_{ll}=\left(\begin{array}{cc}1\\&1\end{array}\right)$
and $\Omega_{jj}=\omega^{-j+1}$ for
$j=2,\cdots,l-1,l+1,\cdots,2l-2$; $K=(K_{jk})_{1\le j,k\le 2l-2}$
with non-zero entries $\D
K_{11}=K_{ll}=\left(\begin{array}{cc}&1\\1\end{array}\right)$ and
$K_{j,2l-j}=1$ for $j=2,\cdots,l-1,l+1,\cdots,2l-2$;
$J=(J_{jk})_{1\le j,k\le 2l-2}$ with non-zero entries $\D
J_{12}=J_{l,l+1}=\left(\begin{array}{c}1\\1\end{array}\right)$,
$J_{l-1,l}=J_{2l-2,1}=\left(\begin{array}{cc}1&1\end{array}\right)$
and $\D J_{j,j+1}=1$ for $j=2,\cdots,l-2,l,\cdots,2l-3$; $\D
V=(V_{jk})_{1\le j,k\le 2l-2}$ with non-zero entries
$V_{11}=\left(\begin{array}{cc}\E{-u_1}\\&\E{u_1}\end{array}\right)$,
$V_{ll}=\left(\begin{array}{cc}\E{u_l}\\&\E{-u_l}\end{array}\right)$,
$V_{jj}=\E{u_j}$ and $V_{2l-j,2l-j}=\E{-u_j}$ for $j=2,\cdots,l-1$.
The equation is
\begin{equation}
   \begin{array}{l}
   \D\triangle u_1=\E{u_2+u_1}-\E{u_2-u_1},\quad
   \D\triangle u_2=\E{u_2+u_1}+\E{u_2-u_1}-\E{u_3-u_2},\\
   \D\triangle u_j=\E{u_j-u_{j-1}}-\E{u_{j+1}-u_j}\quad(j=3,\cdots,l-2),\\
   \D\triangle
   u_{l-1}=\E{u_{l-1}-u_{l-2}}-\E{u_l-u_{l-1}}-\E{-u_l-u_{l-1}},\quad
   \D\triangle u_l=\E{u_l-u_{l-1}}-\E{-u_l-u_{l-1}}.
   \end{array}\label{eq:eqD1}
\end{equation}

By setting $w_j=-(u_1+\cdots+u_j)$ $(j=1,\cdots,l-2)$,
$w_{l-1}=-\frac 12(u_1+\cdots+u_{l-1}+u_l)$, $w_l=-\frac
12(u_1+\cdots+u_{l-1}-u_l)$, it becomes
\begin{equation}
   \begin{array}{l}
   \D\triangle w_1=\E{2w_1-w_2}-\E{-w_2},\\
   \D\triangle w_j=\E{-w_{j-1}+2w_j-w_{j+1}}-2\E{-w_2}\quad(j=2,\cdots,l-3),\\
   \D\triangle w_{l-2}=\E{-w_{l-3}+2w_{l-2}-w_{l-1}-w_l}-2\E{-w_2},\\
   \D\triangle w_{l-1}=\E{-w_{l-2}+2w_{l-1}}-\E{-w_2},\quad
   \D\triangle w_l=\E{-w_{l-2}+2w_l}-\E{-w_2}.
   \end{array}
\end{equation}

Let $(u_1,\cdots,u_l)$ be a solution of (\ref{eq:eqD1}),
$\D\left(\begin{array}{c}h_1\\\vdots\\h_{2l}\end{array}\right)$ be a
column solution of the corresponding Lax pair satisfying
\begin{equation}
   h_{11}h_{12}+(-1)^{l-1}h_{l1}h_{l2}+\sum_{c=2}^{l-1}(-1)^{c-1}h_ch_{2-c}=0
\end{equation}
where $\D
h_1=\left(\begin{array}{c}h_{11}\\h_{12}\end{array}\right)$ and $\D
h_l=\left(\begin{array}{c}h_{l1}\\h_{l2}\end{array}\right)$, and
$h_{11},h_{12},h_2,\cdots,h_{l-1}$, $h_{l1},h_{l2},h_{l+1},\cdots$,
$h_{2l-1}$ are scalar functions. By Theorem~\ref{thm}, the new
solution of (\ref{eq:eqD1}) is
\begin{equation}
   \begin{array}{l}
   \D \widetilde u_1=u_1+\ln\frac{\D|h_{11}|^2\E{u_1}-\gamma_1}
   {\D|h_{12}|^2\E{-u_1}-\gamma_1},\\
   \D\widetilde u_j=u_j+\ln\frac{g_{j+1}}{g_j}\quad
   \D(j=2,\cdots,l-1),\\
   \D \widetilde u_l=u_l+\ln\frac{\D|h_{l1}|^2\E{-u_l}-\gamma_l}
   {\D|h_{l2}|^2\E{u_l}-\gamma_l}
   \end{array}
\end{equation}
where
\begin{equation}
   \begin{array}{rl}
   g_j=&\D\frac1{(2l-2)^2}\Big|2\{1-j\}h_{11}h_{12}
   +2(-1)^{l-1}\{l-j\}h_{l1}h_{l2}\\
   &\D+\sum_{c=2\atop c\ne l}^{2l-2}(-1)^{c+1}
   \{c-j\}h_{2-c}h_c\Big|^2-|\mu|^2\gamma_{j-1}\gamma_{2-j}
   \end{array}
\end{equation}
and
\begin{equation}
   \begin{array}{rl}
   \gamma_j=&\D\frac 1{1-|\mu|^{4l-4}}\Bigg((-|\mu|^2)^{\{j-1\}}
   (|h_{11}|^2\E{u_1}+|h_{12}|^2\E{-u_1})\\
   &\D+(-|\mu|^2)^{\{j+l-2\}}(|h_{l1}|^2\E{-u_l}+|h_{l2}|^2\E{u_l})\\
   &\D+\sum_{c=2}^{l-1}
   (-|\mu|^2)^{\{j+c-2\}}|h_c|^2\E{-u_c}
   +\sum_{c=2}^{l-1}
   (-|\mu|^2)^{\{j-c\}}|h_{2-c}|^2\E{u_c}\Bigg).
   \end{array}
\end{equation}

\section*{Acknowledgements} This work was supported by the
National Basic Research Program of China (973 Program)
(2007CB814800) and the Key Laboratory of Mathematics for Nonlinear
Sciences of Ministry of Education of China.

\end{document}